\documentclass[english]{article}
\usepackage[T1]{fontenc}
\usepackage[latin9]{inputenc}
\usepackage{geometry}
\usepackage{textcomp}
\usepackage{amsthm}
\usepackage{amsmath}
\usepackage{graphicx}
\usepackage{amssymb}
\usepackage[authoryear]{natbib}

\theoremstyle{plain}
\newcommand{\lyxaddress}[1]{
\par {\raggedright #1
\vspace{1.4em}
\noindent\par}
}
\theoremstyle{plain}
\newtheorem{thm}{Theorem}
  \theoremstyle{definition}
  \newtheorem{defn}[thm]{Definition}
  \theoremstyle{remark}
  \newtheorem{rem}[thm]{Remark}
  \theoremstyle{plain}
  \newtheorem{lem}[thm]{Lemma}
  \theoremstyle{plain}
  \newtheorem{cor}[thm]{Corollary}
 \theoremstyle{definition}
  \newtheorem{example}[thm]{Example}

\DeclareMathOperator{\N}{N}

\usepackage{babel}

\begin{document}

\title{Estimating the null distribution for conditional inference and genome-scale
screening }

\author{David R. Bickel}

\maketitle

\lyxaddress{Ottawa Institute of Systems Biology\\
Department of Biochemistry, Microbiology, and Immunology\\
Department of Mathematics and Statistics \\
University of Ottawa\\
451 Smyth Road\\
Ottawa, Ontario, K1H 8M5}

\lyxaddress{}
\begin{abstract}
In a novel approach to the multiple testing problem, Efron (2004;
2007) formulated estimators of the distribution of test statistics
or nominal \emph{p}-values under a null distribution suitable for
modeling the data of thousands of unaffected genes, non-associated
single-nucleotide polymorphisms, or other biological features. Estimators
of the null distribution can improve not only the empirical Bayes
procedure for which it was originally intended, but also many other
multiple comparison procedures. Such estimators serve as the groundwork
for the proposed multiple comparison procedure based on a recent frequentist
method of minimizing posterior expected loss, exemplified with a non-additive
loss function designed for genomic screening rather than for validation. 

The merit of estimating the null distribution is examined from the
vantage point of conditional inference in the remainder of the paper.
In a simulation study of genome-scale multiple testing, conditioning
the observed confidence level on the estimated null distribution as
an approximate ancillary statistic markedly improved conditional inference.
To enable researchers to determine whether to rely on a particular
estimated null distribution for inference or decision making, an information-theoretic
score is provided that quantifies the benefit of conditioning. As
the sum of the degree of ancillarity and the degree of inferential
relevance, the score reflects the balance conditioning would strike
between the two conflicting terms. 

Applications to gene expression microarray data illustrate the methods
introduced.
\end{abstract}

\newpage{}

\section{\label{sec:Introduction}Introduction}

\subsection{\label{sub:Multiplicity}Multiple comparison procedures}

\subsubsection{Aims of multiple-comparison adjustments}

Controversy surrounding whether and how to adjust analysis results
for multiple comparisons can be partly resolved by recognizing that
a procedure that works well for one purpose is often poorly suited
for another since different types of procedures solve distinct statistical
problems. Methods of adjustment have been developed to attain three
goals, the first two of which optimize some measure of sample space
performance:
\begin{enumerate}
\item \emph{Adjustment for selection}. The most common concern leading to
multiple-comparison adjustments stems from the observation that results
can achieve nominal statistical significance because they were selected
to do so rather than because of a reproducible effect. Adjustments
of this type are usually based on control of a Type I error rate such
as a family-wise error rate or a false discovery rate as defined by
\citet{RefWorks:288}. \citet{RefWorks:296} reviewed several options
in the context of gene expression microarray data. 
\item \emph{Minimization of a risk function}. \citet{Stein1956197} proved
that the maximum likelihood estimator (MLE) is inadmissible for estimation
of a multivariate normal mean under squared error loss, even in the
absence of correlation. \citet{EfronMorris1973} extended the result
by establishing that the MLE is dominated by a wide class of estimators
derived via an empirical Bayes approach in which the mean is random.
More recently, \citet{RefWorks:67} adjusted \emph{p}-values for multiple
comparisons by minimizing their risk as estimators of a posterior
probability. In the presence of genome-scale numbers of comparisons,
adjustments based on hierarchical models are often much less extreme
than those needed to adjust for selection. For two examples from microarray
data analysis, \citet{RefWorks:1275} found that posterior intervals
based on a local false discovery rate (LFDR) estimate tend to be substantially
narrower than those needed to control the false coverage rate introduced
by \citet{RefWorks:1276} to account for selection, and an LFDR-based
posterior mean has insufficient shrinkage toward the null to adequately
correct selection bias \citep{RefWorks:26}. 
\item \emph{Estimation of null or alternative distributions.} Measurements
over thousands of biological features available from studies of genome-scale
expression and genome-wide association studies have recently enabled
estimation of distributions of \emph{p}-values. Early empirical Bayes
methods of estimating the LFDR associated with each null hypothesis
employ estimates of the distribution of test statistics or \emph{p}-values
under the alternative hypothesis \citep[e.g.,][]{RefWorks:53}. Efron
\citeyearpar{RefWorks:55,RefWorks:57} went further, demonstrating
the value of also estimating the distribution of \emph{p}-values under
the null hypothesis provided a sufficiently large number of hypotheses
under simultaneous consideration. 
\end{enumerate}
While all three aims are relevant to Neyman-Pearson testing, they
differ as much in their relevance to Fisherian significance testing
as in the procedures they motivate. \citet{RefWorks:573} pointed
out that Type I error rate control is appropriate for making series
of decisions but not for inductive reasoning, where the inferential
evaluation of evidence is of concern apart from loss functions that
depend on how that evidence will be used, which, as \citet[pp. 95-96, 103-106]{RefWorks:985}
stressed, might not even be known at the time of data analysis. Likewise,
\citet{RefWorks:1304} and \citet{RefWorks:1303} found optimization
over the sample space helpful for making series of decisions rather
than for drawing scientific inferences from a particular observed
sample. Cox \citeyearpar{conditionalityPrinciple1958,CoxBook} noted
that selection of a function to optimize is inherently subjective
to the extent that different decision makers have different interests.
Further, sample space optimality is often achieved at the expense
of induction about the parameter given the data at hand; for example,
optimal confidence intervals result from systematically stretching
them in samples of low variance and reducing them in samples of high
variance relative to their conditional counterparts \citep{conditionalityPrinciple1958,RefWorks:1331,RefWorks:1302,FraserAncillaries2004a,FraserAncillaries2004e}. 

The suitability of the methods of both of the first two goals for
decision rules as opposed to inductive reasoning is consistent with
the observation that control of Type I error rates may be formulated
as a minimax problem (e.g., \citealt[§1.5]{RefWorks:1338,Wald1950}),
indicating that the second of the above aims generalizes the first.
Although corrections in order to account for selection are often applied
when it is believed that only a small fraction of null hypotheses
are false \citep{CoxBook}, the methods of controlling a Type I error
rate used to make such corrections are framed in terms of rejection
decisions and thus may depend on the number of tests conducted, which
would not be the case were the degree of correction a function only
of prior beliefs. By contrast with the first two aims, the third aim,
improved specification of the alternative or null distribution of
test statistics, is clearly as important in significance testing as
in fixed-level Neyman-Pearson testing. In short, while the first two
motivations for multiple comparison procedures address decision-theoretic
problems, only the third pertains to significance testing in the sense
of impartially weighing evidence without regard to possible consequences
of actions that might be taken as a result of the findings.

\subsubsection{\label{sub:Estimating-the-null}Estimating the null distribution}

Because of its novelty and its potential importance for many frequentist
procedures of multiple comparisons, the effect of relying on the following
method due to Efron \citeyearpar{RefWorks:55,RefWorks:57,RefWorks:208}
of estimating the null distribution will be examined herein. The method
rests on the assumption that about 90\% or more of a large number
of \emph{p}-values correspond to \emph{unaffected features} and thus
have a common distribution called the \emph{true null distribution}.
If that distribution is uniform, then the \emph{assumed null distribution}
of test statistics with respect to which the \emph{p}-values were
computed is correct. 

In order to model the null distribution as a member of the normal
family, the \emph{p}-values are transformed by $\Phi^{-1}:\left[0,1\right]\rightarrow\mathbb{R}^{1},$
the standard normal quantile function. The parameters of that distribution
are estimated either by fitting a curve to the central region of a
histogram of the transformed \emph{p}-values \citep{RefWorks:55}
or, as used below, by applying a maximum likelihood procedure to a
truncated normal distribution \citep{RefWorks:208}. The main justification
for both algorithms is that since nearly all \emph{p}-values are modeled
as variates from the true null distribution and since the remaining
\emph{p}-values are considered drawn from a distribution with wider
tails, the less extreme \emph{p}-values better resemble the true null
distribution than do those that are more extreme. Since the theoretical
null distribution is standard normal in the transformed domain, deviations
from the standard normal distribution reflect departures in the less
extreme \emph{p}-values from uniformity in the original domain. 

For use in multiple testing, all of the transformed \emph{p}-values
of the data set are treated as test statistics for the derivation
of new \emph{p}-values with respect to the null distribution estimated
as described above instead of the assumed null distribution. Such
adjusted \emph{p}-values would be suitable for inductive inference
or for decision-theoretic analyses such as those controlling error
rates, provided that the true null distribution tends to be closer
to the estimated null distribution than it is to the assumed null
distribution.

\subsection{Overview}

The next section presents a confidence-based distribution of a vector
parameter in order to unify the present study of null distribution
estimation within a single framework. The general framework is then
applied to the problem of estimating the null distribution in Section
\ref{sub:Estimation-of-null}. Section \ref{sub:Non-additive-loss}
introduces a multiple comparisons procedure for coherent decisions
made possible by the confidence-based posterior without recourse to
Bayesian or empirical Bayesian models. 

Adjusting \emph{p}-values by the estimated null distribution is interpreted
as inference conditional on that estimate in Section \ref{sec:Conditional}.
The simulation study of Section \ref{sub:Simulations} demonstrates
that estimation of the null distribution can substantially improve
conditional inference even when the assumed null distribution is correct
marginal over a precision statistic. Section \ref{sub:Merit} provides
a method for determining whether the estimated null distribution is
sufficiently ancillary and relevant for effective conditional inference
or decision making on the basis of a given data set. 

Section \ref{sec:Discussion} concludes with a discussion of the new
findings and methods.

\section{\label{sec:Statistical-framework}Statistical framework}

\subsection{\label{sub:Posterior}Confidence levels as posterior probabilities}

The observed data vector $x\in\Omega$ is modeled as a realization
of a random quantity $X$ of distribution $P_{\xi},$ a probability
distribution on the measurable space $\left(\Omega,\Sigma\right)$
that is specified by the \emph{full parameter} $\xi\in\Xi\subseteq\mathbb{R}^{d}.$
Let $\theta=\theta\left(\xi\right)$ denote a \emph{parameter of interest}
in $\Theta$ and $\gamma=\gamma\left(\xi\right)$ a \emph{nuisance
parameter}. 
\begin{defn}
In addition to the above family of probability measures $\left\{ P_{\xi}:\xi\in\Xi\right\} $,
consider a family of probability measures $\left\{ P^{x}:x\in\Omega\right\} ,$
each on the space $\left(\Theta,\mathcal{A}\right),$ and a set $\mathcal{R}\left(S\right)=\left\{ \hat{\Theta}_{\rho,s\left(\rho\right)}:\rho\in\left[0,1\right],s\in S\right\} $
of region estimators corresponding to a set $S$ of shape functions,
where $\hat{\Theta}_{\rho,s\left(\rho\right)}:\Omega\rightarrow\mathcal{A}$
for all $\rho\in\left[0,1\right]$ and $s\in S.$ If, for every $\Theta^{\prime}\in\mathcal{A},$
$x\in\Omega,$ and $\xi\in\Xi,$ there exist a coverage rate $\rho$
and shape $s\left(\rho\right)$ such that

\begin{equation}
P^{x}\left(\Theta^{\prime}\right)=\rho=P_{\xi}\left(\theta\left(\xi\right)\in\hat{\Theta}_{\rho,s\left(\rho\right)}\left(X\right)\right)\label{eq:matching}\end{equation}
and $\hat{\Theta}_{\rho,s\left(\rho\right)}\left(x\right)=P^{x}\left(\Theta^{\prime}\right),$
then the probability $P^{x}\left(\Theta^{\prime}\right)$ is the \emph{confidence
level} of the hypothesis $\theta\left(\xi\right)\in\Theta^{\prime}$
according to $P^{x},$ the \emph{confidence measure} over $\Theta$
corresponding to $\mathcal{R}\left(S\right).$ \end{defn}
\begin{rem}
Unless the $\sigma$-field $\mathcal{A}$ is Borel, the confidence
level of the hypothesis of interest will not necessarily be defined;
cf. \citet{FiducialMcCullagh2004}.
\end{rem}
Building on work of \citet{RefWorks:249} and others, \citet{Polansky2007b}
used the equivalent of $P^{x}$ to concisely communicate a distribution
of {}``observed confidence'' or {}``attained confidence'' levels
for each hypothesis that $\theta$ lies in some region $\Theta^{\prime}.$
The decision-theoretic {}``certainty'' interpretation of $P^{x}$
as a non-Bayesian posterior \citep{CoherentFrequentism} serves the
same purpose but also ensures the coherence of actions that minimize
expected posterior loss. \citet{RefWorks:1081} also considered interpreting
the ratio $\rho/\left(1-\rho\right)$ from equation (\ref{eq:matching})
as odds for betting on the hypothesis that $\theta\in\Theta^{\prime}.$ 

The posterior distribution need not conform to the Bayes update rule
\citep{CoherentFrequentism} since decisions that minimize posterior
expected loss, or, equivalently, maximize expected utility, are coherent
as long as the posterior distribution is some finitely additive probability
distribution over parameter space (see, e.g., \citealp{RefWorks:1488}).
It follows that an intelligent agent that acts as if $\rho/\left(1-\rho\right)$
are fair betting odds for the hypothesis that $\theta$ lies in a
level-$\rho$ confidence region estimated by some region estimator
of exact coverage rate $\rho$ is coherent if and only if its actions
minimize expected loss with the expectation value over a confidence
measure as the probability distribution defining the expectation value
(cf. \citealp{CoherentFrequentism}). Minimizing expected loss over
the parameter space, whether based on a confidence posterior or on
a Bayesian posterior, differs fundamentally from the decision-theoretic
approach of Section \ref{sub:Multiplicity} in that the former is
optimal given the single sample actually observed whereas the latter
is optimal over repeated sampling. Section \ref{sub:Non-additive-loss}
illustrates the minimization of confidence-measure expected loss with
an application to screening  on the basis of genomics data.

\subsection{\label{sub:Two-sided-testing}Confidence levels versus \emph{p}-values}

Whether confidence levels agree with \emph{p}-values depends on the
parameter of interest and on the chosen hypotheses. If $\theta$ is
a scalar and the null hypothesis is $\theta=\theta^{\prime},$ the
\emph{p}-values associated with the alternative hypotheses $\theta>\theta^{\prime}$
and $\theta<\theta^{\prime}$ are $P^{x}\left(\left(-\infty,\theta^{\prime}\right)\right)$
and $P^{x}\left(\left(\theta^{\prime},\infty\right)\right),$ respectively;
cf. \citet{RefWorks:127}. 

On the other hand, a \emph{p}-value associated with a two-sided alternative
is not typically equal to the confidence level $P^{x}\left(\left\{ \theta^{\prime}\right\} \right).$
\citet[pp. 126-128, 216]{Polansky2007b} discusses the tendency of
the attained confidence level of a point or simple hypotheses such
as $\theta=\theta^{\prime}$ to vanish in a continuous parameter space.
That only a finite number of points in hypothesis space have nonzero
confidence is required of any evidence scale that is fractional in
the sense that the total strength of evidence over $\Theta$ is finite.
(Fractional scales enable statements of the form, {}``the negative,
null, and positive hypotheses are 80\%, 15\%, and 5\% supported by
the data, respectively.'') While the usual two-sided \emph{p}-value
vanishes only for sufficiently large samples, the confidence level
$P^{x}\left(\left\{ \theta^{\prime}\right\} \right)$ typically is
0\% even for the smallest samples and thus does not lead to the appearance
of a paradox of {}``over-powered'' studies. As a remedy, \citet{HodgesLehmann1954}
proposed testing an interval hypothesis $\theta\in\Theta^{\prime}$
defined in terms of scientific significance; in this situation, as
with composite hypothesis-testing in general, $P^{x}\left(\Theta^{\prime}\right)$
converges in probability to $1_{\Theta^{\prime}}\left(\theta\right)$
even though the two-sided \emph{p}-value does not \citep{CoherentFrequentism}.
(Testing a simple null hypothesis against a composite alternative
hypothesis yields a similar discrepancy between a two-sided \emph{p}-value
and methods that respect the likelihood principle \citep{Levine1996331,RefWorks:435}.)

There are nonetheless situations that, when using \emph{p}-values
for statistical significance, necessitate testing a hypothesis known
to be false for all practical purposes. \citet{RefWorks:294} called
a null hypothesis $\theta=\theta^{\prime}$ \emph{dividing} if it
is not considered because it could possibly be approximately true
but rather because it lies on the boundary between $\theta<\theta^{\prime}$
and $\theta>\theta^{\prime}$, the two hypotheses of genuine interest.
For example, a test of equality of means and its associated two-sided
\emph{p}-value often serve the purpose of determining whether there
are enough data to determine the direction of the difference when
it is known that there is some appreciable difference \citep{RefWorks:294}.
That goal can be more directly attained by comparing the confidence
levels $P^{x}\left(\left(-\infty,\theta^{\prime}\right)\right)$ and
$P^{x}\left(\left(\theta^{\prime},\infty\right)\right).$ While reporting
the ratio or maximum of $P^{x}\left(\left(-\infty,\theta^{\prime}\right)\right)$
and $P^{x}\left(\left(\theta^{\prime},\infty\right)\right)$ would
summarize the confidence level of each of two regions in a single
number, such a number may be more susceptible to misinterpretation
than a report of the pair of confidence levels.

\subsection{\label{sub:Simultaneous}Simultaneous inference}

In the typical genome-scale problem, there are $d$ scalar parameters
$\theta_{1},$ $\theta_{2},$ ..., $\theta_{d}$ and $d$ corresponding
observables $X_{1},$ $X_{2},$ ..., $X_{d},$ such that $d\ge1000$
and $\theta_{i}=\theta_{i}\left(\xi\right)$ is a subparameter of
the distribution of $X_{i},$ the random quantity of which the observation
$x_{i}\in\Omega_{i}$ is a realized vector. The $i$th of the $d$
hypotheses to be simultaneously tested is $\theta_{i}\in\Theta_{i}^{\prime}$
for some $\Theta_{i}^{\prime}$ in $\Theta_{i},$ a subset of $\mathbb{R}^{1}.$
Representing numeric tuples under the angular bracket convention to
distinguish the open interval $\left(x,y\right)$ from the ordered
pair $\left\langle x,y\right\rangle ,$ $\theta=\theta\left(\xi\right)=\left\langle \theta_{1},\theta_{2},\ldots,\theta_{d}\right\rangle $
is the $d$-dimensional subparameter of interest and the joint hypothesis
is $\theta\left(\xi\right)\in\Theta^{\prime},$ where $\Theta^{\prime}=\Theta_{1}^{\prime}\times\Theta_{2}^{\prime}\times\cdots\times\Theta_{d}^{\prime}$.

For any $\delta\in\left\{ 1,2,...,d-1\right\} ,$ inference may focus
on $\delta$ of the scalar parameters as opposed to the entire vector
$\theta.$ For example, separate consideration of the confidence levels
of hypotheses such as $\theta_{1}\in\Theta_{1}^{\prime}$ or of $\left\langle \theta_{1},\theta_{2}\right\rangle \in\Theta_{1}^{\prime}\times\Theta_{2}^{\prime}$
can be informative, especially if $d$ is high. Each of the components
of the \emph{focus index }$\iota=\left\langle i\left(1\right),i\left(2\right),\ldots,i\left(\delta\right)\right\rangle $
is in $\left\{ 1,...,d\right\} $ and is unequal to each of its other
components. The proper subset $\tilde{\Theta}_{\iota}^{\prime}=\Theta_{i\left(1\right)}^{\prime}\times\Theta_{i\left(2\right)}^{\prime}\times\cdots\times\Theta_{i\left(\delta\right)}^{\prime}$
of $\tilde{\Theta}_{\iota}=\Theta_{i\left(1\right)}\times\Theta_{i\left(2\right)}\times\cdots\times\Theta_{i\left(\delta\right)}$
is defined in order to weigh the evidence for the hypothesis that
$\tilde{\theta}_{\iota}=\left\langle \theta_{i\left(1\right)},\theta_{i\left(2\right)},\ldots,\theta_{i\left(\delta\right)}\right\rangle \in\tilde{\Theta}_{\iota}^{\prime}.$
Setting $\Theta_{\iota}^{\prime}=\Theta_{1}^{\prime}\times\Theta_{2}^{\prime}\times\cdots\times\Theta_{d}^{\prime}$
such that $\Theta_{j}^{\prime}=\Theta_{j}$ for all $j\notin\left\{ i\left(1\right),i\left(2\right),\ldots,i\left(\delta\right)\right\} ,$
define the marginal distribution $P_{\iota}^{x}$ such that $P_{\iota}^{x}\left(\tilde{\Theta}_{\iota}^{\prime}\right)$
is equal to the confidence level $P^{x}\left(\Theta_{\iota}^{\prime}\right).$
Thus, $P_{\iota}^{x}$ is a probability measure marginal over all
$\theta_{j}$ with $j\notin\left\{ i\left(1\right),i\left(2\right),\ldots,i\left(\delta\right)\right\} .$ 

The following lemma streamlines inference focused on whether $\tilde{\theta}_{\iota}\in\tilde{\Theta}_{\iota}^{\prime},$
or, equivalently, $\theta\left(\xi\right)\in\Theta_{\iota}^{\prime},$
by establishing sufficient conditions for the confidence level marginal
over some of the $d$ components of $\theta$ to be equal to the parameter
coverage probability marginal over the data corresponding to those
components. 
\begin{lem}
Considering a focus index $\iota$ and $\tilde{X}_{\iota}=\left\langle X_{i\left(1\right)},X_{i\left(2\right)},\ldots,X_{i\left(\delta\right)}\right\rangle ,$
let $\hat{\Theta}_{\rho,s\left(\rho\right)}^{\iota}:\Omega\rightarrow\tilde{\mathcal{A}}_{\iota}$
be the corresponding level-$\rho$ set estimator of some shape parameter
$s\left(\rho\right)$ defined such that for every $x\in\Omega,$ $\hat{\Theta}_{\rho,s\left(\rho\right)}^{\iota}\left(x\right)$
is the canonical projection of $\hat{\Theta}_{\rho,s\left(\rho\right)}\left(x\right)$
from $\mathcal{A}$ to $\tilde{\mathcal{A}}_{\iota},$ the $\sigma$-field
of the marginal distribution $P_{\iota}^{x}.$ If there is a map $\tilde{\Theta}_{\rho,s\left(\rho\right)}^{\iota}:\tilde{\Omega}_{\iota}\rightarrow\tilde{\mathcal{A}}_{\iota}$
such that $\tilde{\Theta}_{\rho,s\left(\rho\right)}^{\iota}\left(\tilde{X}_{\iota}\right)$
and $\hat{\Theta}_{\rho,s\left(\rho\right)}^{\iota}\left(X\right)$
are identically distributed, then $P_{\iota}^{x}$ is the confidence
measure over $\tilde{\Theta}_{\iota}$ corresponding to $\left\{ \tilde{\Theta}_{\rho,s\left(\rho\right)}^{\iota}:\rho\in\left[0,1\right],s\in S\right\} .$
\end{lem}
\begin{proof}
By the general definition of confidence level (\ref{eq:matching}),
\[
P_{\iota}^{x}\left(\tilde{\Theta}_{\iota}^{\prime}\right)=P^{x}\left(\Theta_{\iota}^{\prime}\right)=P_{\xi}\left(\theta\in\hat{\Theta}_{\rho,s\left(\rho\right)}\left(X\right)\right),\]
where the coverage rate $\rho$ and shape parameter $s\left(\rho\right)$
are constrained such that $\hat{\Theta}_{\rho,s\left(\rho\right)}\left(x\right)=\Theta_{\iota}^{\prime}$
for the observed value $x$ of random element $X.$ Hence, using $A_{\iota}$
to denote the event that $\theta_{j}\in\Theta_{j}^{\prime}$ ,\begin{equation}
P_{\iota}^{x}\left(\tilde{\Theta}_{\iota}^{\prime}\right)=P_{\xi}\left(\tilde{\theta}_{\iota}\in\hat{\Theta}_{\rho,s\left(\rho\right)}^{\iota}\left(X\right),A_{\iota}\right)\label{eq:marginal-weight-is-probability}\end{equation}
with the coverage rate $\rho$ and shape parameter $s\left(\rho\right)$
restricted such that $\hat{\Theta}_{\rho,s\left(\rho\right)}^{\iota}\left(x\right)=\tilde{\Theta}_{\iota}^{\prime}.$
Considering $j\notin\left\{ i\left(1\right),i\left(2\right),\ldots,i\left(\delta\right)\right\} ,$
the event $A_{\iota}$ satisfies $P_{\xi}\left(A_{\iota}\right)=1$
since $\Theta_{j}^{\prime}=\Theta_{j},$ thereby eliminating $A_{\iota}$
from equation (\ref{eq:marginal-weight-is-probability}). Because
$\tilde{\Theta}_{\rho,s\left(\rho\right)}^{\iota}$ exists by assumption,
$\tilde{\Theta}_{\rho,s\left(\rho\right)}^{\iota}\left(\tilde{x}_{\iota}\right)=\tilde{\Theta}_{\iota}^{\prime}$
results and $\tilde{\Theta}_{\rho,s\left(\rho\right)}^{\iota}\left(\tilde{X}_{\iota}\right)$
replaces $\hat{\Theta}_{\rho,s\left(\rho\right)}^{\iota}\left(X\right)$
in equation (\ref{eq:marginal-weight-is-probability}) since they
are identically distributed. Therefore, \[
P_{\iota}^{x}\left(\tilde{\Theta}_{\iota}^{\prime}\right)=\rho=P_{\xi}\left(\tilde{\theta}_{\iota}\in\tilde{\Theta}_{\rho,s\left(\rho\right)}^{\iota}\left(\tilde{X}_{\iota}\right)\right),\]
where the coverage rate $\rho$ and shape parameter $s\left(\rho\right)$
are constrained such that $\tilde{\Theta}_{\rho,s\left(\rho\right)}^{\iota}\left(\tilde{x}_{\iota}\right)=\tilde{\Theta}_{\iota}^{\prime}$
for the observed value $\tilde{x}_{\iota}=\left\langle x_{i\left(1\right)},x_{i\left(2\right)},\ldots,x_{i\left(\delta\right)}\right\rangle $
of $\tilde{X}_{\iota}.$
\end{proof}
Conditional independence is sufficient to satisfy the lemma's condition
of identically distributed region estimators:
\begin{thm}
If $X_{i}$ is conditionally independent of $X_{j}$ and $\theta_{j}$
given $\theta_{i}$ for all $i\ne j,$ then, for any focus index $\iota,$
there is a map $\tilde{\Theta}_{\rho,s\left(\rho\right)}^{\iota}:\tilde{\Omega}_{\iota}\rightarrow\tilde{\mathcal{A}}_{\iota}$
such that $\tilde{\Theta}_{\rho,s\left(\rho\right)}^{\iota}\left(\tilde{x}_{\iota}\right)=\hat{\Theta}_{\rho,s\left(\rho\right)}^{\iota}\left(x\right)$
with $\tilde{x}_{\iota}=\left\langle x_{i\left(1\right)},x_{i\left(2\right)},\ldots,x_{i\left(\delta\right)}\right\rangle $
for every $x\in\Omega,$ and the marginal distribution $P_{\iota}^{x}$
is the confidence measure over $\tilde{\Theta}_{\iota}$ corresponding
to $\left\{ \tilde{\Theta}_{\rho,s\left(\rho\right)}^{\iota}:\rho\in\left[0,1\right],s\in S\right\} .$\end{thm}
\begin{proof}
By the conditional independence assumption, $\hat{\Theta}_{\rho,s\left(\rho\right)}^{\iota}\left(X\right)$
is conditionally independent of $\theta_{j}$ and $X_{j}$ for all
$j\notin\left\{ i\left(1\right),i\left(2\right),\ldots,i\left(\delta\right)\right\} $
given $\tilde{\theta}_{\iota},$ entailing the existence of a map
$\tilde{\Theta}_{\rho,s\left(\rho\right)}^{\iota}:\tilde{\Omega}_{\iota}\rightarrow\tilde{\mathcal{A}}_{\iota}$
such that $\tilde{\Theta}_{\rho,s\left(\rho\right)}^{\iota}\left(\tilde{X}_{\iota}\right)$
and $\hat{\Theta}_{\rho,s\left(\rho\right)}^{\iota}\left(X\right)$
are identically distributed. Then the above lemma yields the consequent.
\end{proof}
The theorem facilitates inference separately focused on each scalar
subparameter $\theta_{i}$ on the basis of the observation that $X_{i}=x_{i}\in\Omega_{i}$:
\begin{cor}
\label{cor:independence}If $X_{i}$ is conditionally independent
of $X_{j}$ and $\theta_{j}$ given $\theta_{i}$ for all $i\ne j,$
then, for any $i\in\left\{ 1,2,\dots,k\right\} ,$ the marginal distribution
$P_{\left\langle i\right\rangle }^{x}$ is the confidence measure
over $\Theta_{i}$ corresponding to some set $\left\{ \tilde{\Theta}_{\rho,s\left(\rho\right)}^{\left\langle i\right\rangle }:\rho\in\left[0,1\right],s\in S\right\} $
of interval estimators, each a map $\tilde{\Theta}_{\rho,s\left(\rho\right)}^{\left\langle i\right\rangle }:\Omega_{i}\rightarrow\tilde{\mathcal{A}}_{\left\langle i\right\rangle }.$\end{cor}
\begin{proof}
Under the stated conditions, the theorem entails the existence of
a map $\tilde{\Theta}_{\rho,s\left(\rho\right)}^{\left\langle i\right\rangle }:\tilde{\Omega}_{\left\langle i\right\rangle }\rightarrow\tilde{\mathcal{A}}_{\left\langle i\right\rangle }$
such that $\tilde{\Theta}_{\rho,s\left(\rho\right)}^{\left\langle i\right\rangle }\left(\tilde{x}_{\left\langle i\right\rangle }\right)=\hat{\Theta}_{\rho,s\left(\rho\right)}^{\left\langle i\right\rangle }\left(x\right)$
with $\tilde{x}_{\left\langle i\right\rangle }=x_{i}$ for every $x\in\Omega$
and entails that the marginal distribution $P_{\left\langle i\right\rangle }^{x}$
is the confidence measure over $\tilde{\Theta}_{\left\langle i\right\rangle }$
corresponding to $\left\{ \tilde{\Theta}_{\rho,s\left(\rho\right)}^{\left\langle i\right\rangle }:\rho\in\left[0,1\right],s\in S\right\} .$\end{proof}
\begin{rem}
\label{rem:marginal-over-components}The applications of Sections
\ref{sec:Null-estimation} and \ref{sec:Conditional} exploit this
property in order to draw inferences from the confidence levels $P_{\left\langle 1\right\rangle }^{x}\left(\left(\inf\Theta_{1},\theta^{\prime}\right)\right),P_{\left\langle 2\right\rangle }^{x}\left(\left(\inf\Theta_{2},\theta^{\prime}\right)\right),\ldots,P_{\left\langle d\right\rangle }^{x}\left(\left(\inf\Theta_{d},\theta^{\prime}\right)\right)$
of the hypotheses $\theta_{1}<\theta^{\prime},$ $\theta_{2}<\theta^{\prime},$
..., $\theta_{d}<\theta^{\prime}$, respectively, for very large $d$.
Here, $\delta=1,$ each subscript $\left\langle j\right\rangle $
is the 1-tuple representation of the vector $\iota$ with $j$ as
its only component, and $\theta^{\prime}$ is the scalar supremum
shared by all $d$ hypotheses. 
\end{rem}

\section{\label{sec:Null-estimation}Null estimation for genome-scale screening}

\subsection{\label{sub:Estimation-of-null}Estimation of the null posterior}

In the presence of hundreds or thousands of hypotheses, the novel
methodology of \citet{RefWorks:57} can improve evidential inference
by estimation of the null distribution. While \citet{RefWorks:57}
originally applied the estimator to effectively condition the LFDR
on an estimated distribution of null \emph{p}-values, he noted that
its applications potentially encompass any procedure that depends
on the distribution of statistics under the null hypothesis. Indeed,
the stochasticity of parameters that enables estimation of the LFDR
by the empirical Bayes machinery need not be assumed for the pre-decision
purpose of deriving the level of confidence that each gene is differentially
expressed. Thus, the methodology of \citet{RefWorks:57} outlined
in Section in terms of \emph{p}-values can be appropriated to adjust
confidence levels (§\ref{sec:Statistical-framework}) since $P_{\left\langle i\right\rangle }^{x}\left(\left(-\infty,\theta^{\prime}\right)\right),$
the level of confidence that a scalar subparameter $\theta_{i}$ is
less than a given scalar $\theta^{\prime},$ is numerically equal
to $p_{\left\langle i\right\rangle }^{x}\left(\theta^{\prime}\right),$
the upper-tailed \emph{p}-value for the test of the hypothesis that
$\theta_{i}=\theta^{\prime}.$ Specifically, confidence levels are
adjusted in this paper according to the estimated confidence measure
under the null hypothesis rather than according to an assumed confidence
measure under the null hypothesis. 

Treating the parameters indicating differential expression as fixed
rather than as exchangeable random quantities arguably provides a
closer fit to the biological system in the sense that certain genes
remain differentially expressed and other genes remain by comparison
equivalently expressed across controlled conditions under repeated
sampling. While the confidence measure is a probability measure on
parameter space, its probabilities are interpreted as a degrees of
confidence suitable for coherent decision making (§\ref{sub:Non-additive-loss}),
not as physical probabilities modeling a frequency of events in the
system. The interpretation of parameter randomness in LFDR methods
is less clear except when the LFDR is seen as an approximation to
a Bayesian posterior probability under a hierarchical model. 
\begin{example}
\label{exa:tomato}A tomato development experiment of \citet{RefWorks:8}
yielded $n=6$ observed ratios of mutant expression to wild-type expression
in most of the $d=13,340$ genes on the microarray with missing data
for many genes. For the $i$th gene, the interest parameter $\theta_{i}$
is the expectation value of $X_{i},$ the logarithm of the expression
ratio. The hypothesis $\theta_{i}<0$ corresponds to downregulation
of gene $i$ in the mutant, whereas $\theta_{i}>0$ corresponds to
upregulation. To obviate estimation of a joint distribution of $d$
parameters, the independence conditions of Corollary \ref{cor:independence}
are assumed to hold. Also assuming normally distributed $X_{i},$
the one-sample \emph{t}-test gave the upper-tail \emph{p}-value equal
to the confidence level $P_{\left\langle i\right\rangle }^{x_{i}}\left(\mathbb{R}_{-}\right)$
for each gene. The notation is that of Remark \ref{rem:marginal-over-components},
except with the replacement of each $x$ subscript with $x_{i}$ to
emphasize that only the $i$th observed vector influences the confidence
level corresponding to the $i$th parameter. Efron's \citeyearpar{RefWorks:208}
maximum-likelihood method of estimating the null distribution from
a vector of \emph{p}-values provided the estimated null confidence
measure that is very close to the empirical distribution of the data
(Fig. \ref{fig:CDFs}), which is consistent with but does not imply
the truth of all null hypotheses of equivalent expression $\left(\theta_{i}=0\right)$.
Using that estimate of the null distribution in place of the uniform
distribution corresponding to the Student \emph{t }distribution of
test statistics has the effect of adjusting each confidence level.
Since extreme confidence levels are adjusted toward 1/2, the estimated
null reduces the confidence level both of genes with large values
of $P_{\left\langle i\right\rangle }^{x_{i}}\left(\mathbb{R}_{-}\right)$
(confidence of the hypothesis $\theta_{i}<0$) and of those with large
values of $P_{\left\langle i\right\rangle }^{x_{i}}\left(\mathbb{R}_{+}\right)$
(confidence of the hypothesis $\theta_{i}>0$). Fig. \ref{fig:microarray}
displays the effect of this confidence-level adjustment in more detail. 
\end{example}
\begin{figure}
\includegraphics[scale=0.7]{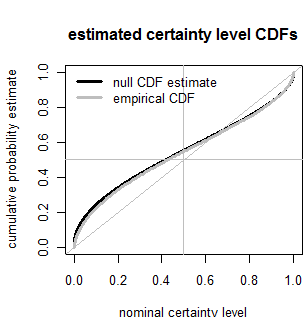}

\caption{The black curve is the estimated cumulative distribution function
(CDF) of the confidence levels under the null distribution, which
corresponds to equivalently expressed or unaffected genes; the gray
curve is the empirical CDF of all confidence levels, including those
of differentially expressed or affected genes. Here, observed confidence
coefficients corresponding to hypotheses are interpreted as levels
of certainty (§§\ref{sub:Posterior}, \ref{sub:Non-additive-loss}).
Departure of the black curve from the diagonal line reflects violation
of independence or of the lognormal assumption used to compute the
confidence levels. As one-sided \emph{p}-values, these confidence
levels would be uniform under the hypothesis of equivalent expression
given the assumptions; i.e., the $\Phi^{-1}$-transformed confidence
levels of unaffected genes are assumed to be $\N\left(0,1^{2}\right),$
where $\Phi^{-1}$ is the standard normal quantile function. The distribution
of $\Phi^{-1}$-transformed confidence levels under that null hypothesis
was estimated to instead be $\N\left(-0.21,\left(1.55\right)^{2}\right).$
The data set, model, and null distribution estimator are those of
Example \ref{exa:tomato}.\label{fig:CDFs}}

\end{figure}

\begin{figure}
\includegraphics[scale=0.55]{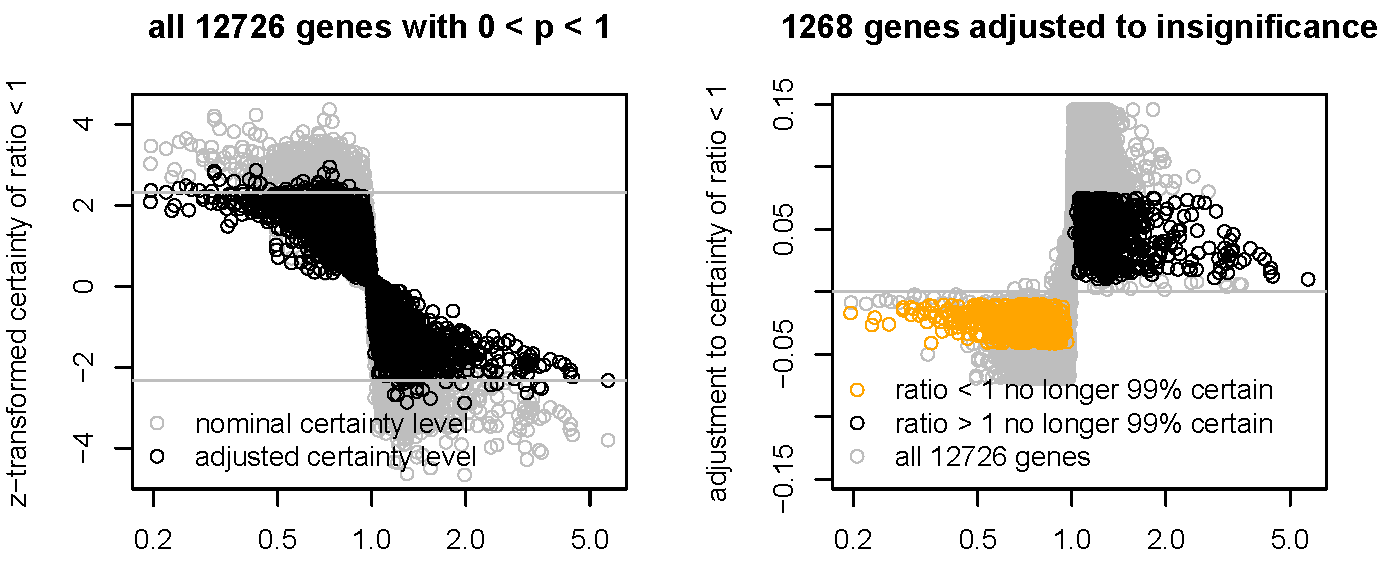}

\caption{Impact of null estimation on the confidence level as the measure of
certainty or statistical significance. The data set, model, and null
distribution estimator are those of Example \ref{exa:tomato} and
Fig. \ref{fig:CDFs}. \emph{Left panel:} The transformed confidence
level $\Phi^{-1}\left(P_{\left\langle i\right\rangle }^{x_{i}}\left(\mathbb{R}_{-}\right)\right)$
for gene $i$ versus the expression ratio estimated as the geometric
sample mean of the observed expression ratio for the same gene. Here,
the confidence level $P_{\left\langle i\right\rangle }^{x_{i}}\left(\mathbb{R}_{-}\right)$
is the degree of certainty of the hypothesis that the mean log-transformed
expression ratio is negative or, equivalently, of the hypothesis that
the true expression ratio is less than 1. The horizontal lines are
drawn at $P_{\left\langle i\right\rangle }^{x_{i}}\left(\mathbb{R}_{-}\right)=99\%$
and at $P_{\left\langle i\right\rangle }^{x_{i}}\left(\mathbb{R}_{+}\right)=1-P_{\left\langle i\right\rangle }^{x_{i}}\left(\mathbb{R}_{-}\right)=99\%.$
Of the original 13,340 genes, 1062 genes have less than the two observations
needed for the test statistic and 2 genes have infinite normal-transformed
confidence levels and thus are not displayed. Each circle corresponds
to a gene, with black for $P_{\left\langle i\right\rangle }^{x_{i}}\left(\mathbb{R}_{-};\hat{F}_{0}\right),$
the confidence level of $\theta_{i}\in\mathbb{R}_{-}$ using the estimated
null distribution $\hat{F}_{0}$ and with gray for $P_{\left\langle i\right\rangle }^{x_{i}}\left(\mathbb{R}_{-};\tilde{F}_{0}\right),$
the same except using the assumed null distribution $\tilde{F}_{0}.$
\emph{Right panel:} The difference between $P_{\left\langle i\right\rangle }^{x_{i}}\left(\mathbb{R}_{-};\hat{F}_{0}\right)$
and $P_{\left\langle i\right\rangle }^{x_{i}}\left(\mathbb{R}_{-};\tilde{F}_{0}\right)$
versus the estimated expression ratio. Orange circles represent genes
satisfying $P_{\left\langle i\right\rangle }^{x_{i}}\left(\mathbb{R}_{-};\tilde{F}_{0}\right)>99\%$
but $P_{\left\langle i\right\rangle }^{x_{i}}\left(\mathbb{R}_{-};\hat{F}_{0}\right)\le99\%$;
black circles represent genes satisfying $P_{\left\langle i\right\rangle }^{x_{i}}\left(\mathbb{R}_{+};\tilde{F}_{0}\right)>99\%$
but $P_{\left\langle i\right\rangle }^{x_{i}}\left(\mathbb{R}_{+};\hat{F}_{0}\right)\le99\%.$
\label{fig:microarray}}

\end{figure}

\subsection{\label{sub:Non-additive-loss}Genome-scale screening loss}

\citet[§B.5.2]{citeulike:2758749} observed that with a suitable non-additive
loss function, optimal decisions in the presence of multiple comparisons
can be made on the basis of minimizing posterior expected loss. A
simple non-additive loss function is \begin{equation}
L_{a,c}\left(M,m\right)=cM^{1+a}+m,\label{eq:loss}\end{equation}
where $M$ and $m$ are respectively the number of incorrect decisions
and the number of non-decisions concerning the $d$ components of
$\theta;$ $M+m\le d.$ The scalars $a\in\mathbb{R}^{1}$ and $c>0$
reflect different aspects of risk aversion: $a$ is an acceleration
in the sense of quantifying the interactive compounding effect of
multiple errors, whereas if $a=0,$ then $c$ is the ratio of the
cost of making an incorrect decision to the cost of not making any
decision or, equivalently, the benefit of making a correct decision. 

\citet{RefWorks:22} and \citet{RefWorks:111} applied additive loss
$\left(a=0\right)$ to decisions of whether or not a biological feature
is affected. That special case, however, does not accurately represent
the screening purpose of most genome-scale studies, which is to formulate
a reasonable number of hypotheses about features for confirmation
in a follow-up experiment. More suitable for that goal, $a>0$ allows
generation of hypotheses for at least a few features even on slight
evidence without leading to unmanageably high numbers of features
even in the presence of decisive evidence. 

Fig. \ref{fig:decisions} displays the result of minimizing such an
expected loss with respect to the confidence posterior (\ref{eq:matching})
under the above class of loss functions (\ref{eq:loss}) for decisions
on the direction of differential gene expression (Example \ref{exa:tomato}).
(Taking the expectation value over the confidence measure rather than
over a Bayesian posterior measure was justified in Section \ref{sub:Posterior}.)

\begin{figure}
\includegraphics[scale=0.7]{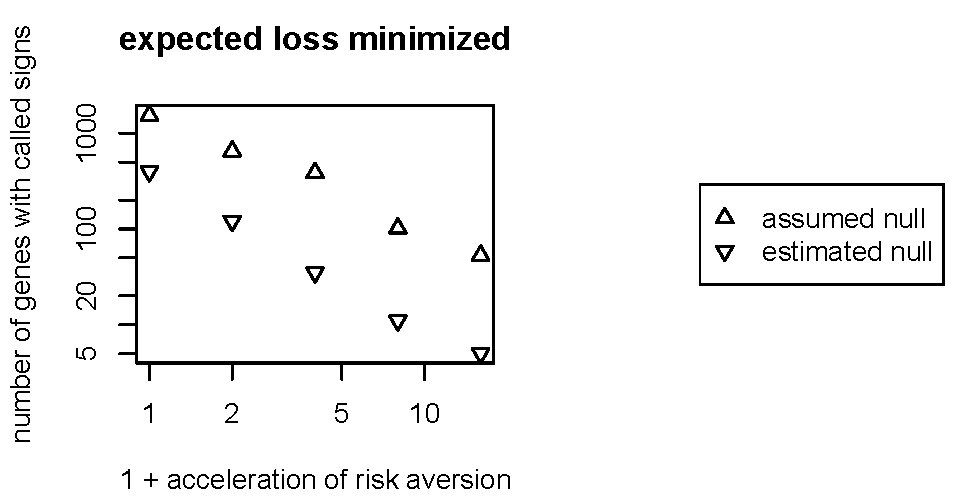}

\caption{Number $d-m$ of decisions on whether the $i$th gene is overexpressed
$\left(\theta_{i}>0\right)$ or underexpressed $\left(\theta_{i}<0\right)$
plotted against $1+a,$ where $a$ is the degree to which the loss
per incorrect decision increases with the number of incorrect decisions
(\ref{eq:loss}). The sign call or decision on the direction of regulation
for each gene was either made or not made such that the following
Monte Carlo approximation to the expected loss $E^{x}\left(L_{a,9}\left(M,m\right)\right)=\int L_{a,9}\left(M,m\right)dP^{x}$
was minimized based alternately on the assumed null distribution $\tilde{F}_{0}$
and on the estimated null distribution $\hat{F}_{0}.$ The $k$th
of the $10^{4}$ values of $\theta_{i}$ was drawn from the frequentist
posterior (\ref{eq:matching}) independently for each gene $i$ to
compute the correct sign decisions according to the $k$th realization;
such correct decisions yielded $M_{k}$ and $m_{k},$ the number of
incorrect sign decisions and the number of non-decisions. The independence
of $\sigma$-fields corresponding to each gene's scalar component
of $\theta$ guaranteed by Corollary \ref{cor:independence} implies
$E^{x}\left(L_{a,9}\left(M,m\right)\right)\doteq10^{-4}\sum_{k=1}^{10^{4}}L_{a,9}\left(M_{k},m_{k}\right).$
The data set, model, and null distribution estimator are those of
Example \ref{exa:tomato} and Figs. \ref{fig:CDFs} and  \ref{fig:microarray}.\label{fig:decisions}}

\end{figure}

\section{\label{sec:Conditional}Null estimation as conditional inference}

\subsection{\label{sub:Simulations}Simulation study}

To record the effect of null distribution estimation on inductive
inference, a simulation study was conducted with $K=500$ independent
samples each of $d=10,000$ independent observable vectors, of which
95\% correspond to unaffected and 5\% to affected features such as
genes or single-nucleotide polymorphisms (SNPs). In Example \ref{exa:tomato},
an affected gene is one for which there is differential gene expression
between mutant and wild type. Assuming that each scalar parameter
$\theta_{i}$ is constrained to lie in the same set $\Theta_{1},$
the one-sided \emph{p}-value of each observable is equal to $P_{k,i}^{x}\left(\left(\inf\Theta_{1},\theta^{\prime}\right)\right),$
the $k$th confidence level of $\theta_{i}<\theta^{\prime},$ the
hypothesis that the parameter of interest for the $i$th observable
vector or feature is less than some value $\theta^{\prime}$ dividing
two meaningful hypotheses, as discussed in Section \ref{sub:Two-sided-testing}
and illustrated in Fig. \ref{fig:microarray}. (This notation differs
from that of Remark \ref{rem:marginal-over-components} in adapting
the superscript of the confidence level and from that of Example \ref{exa:tomato}
in dropping the subscript of $x_{k,i}$ for ease of reading.) As $\theta_{i}=\theta^{\prime}$
is treated as a null hypothesis for the purpose of estimating or assuming
the null distribution, it naturally corresponds an unaffected feature.
Each confidence level was generated from $\Phi,$ the standard normal
CDF, of $Z_{k,i}\sim\N\left(0,\varsigma_{k}^{2}\right)$ for $i\in\left\{ 1,\ldots,9500\right\} $
or of $Z\sim\N\left(5\varsigma_{k}/2,\left(5\varsigma_{k}/4\right)^{2}\right)$
for $i\in\left\{ 9501,\ldots,10^{4}\right\} $. Rather than fixing
$\varsigma_{k}$ at 1 for all $k$ \citep[Fig. 5]{RefWorks:57}, $\varsigma_{k}$
was instead allowed to vary across samples in order to model sample-specific
variation that influences the distribution of \emph{p}-values. For
every $k$ in $\left\{ 1,\ldots,K\right\} ,$ $\log\varsigma_{k}$
is independent and equal to 2/3 with probability 30\%, 1 with probability
40\%, or 3/2 with probability 30\%. Each simulated sample was analyzed
with the same maximum-likelihood method of estimating the null distribution
used in the above gene expression example, in which the realized value
of $\varsigma_{k}$ was predicted to be about 3/2 (Fig. \ref{fig:CDFs}).

Because $\varsigma_{k}$ is an ancillary statistic in the sense that
its distribution is not a function of the parameter and since estimation
of the null distribution approximates conditioning the \emph{p}-values
and equivalent confidence levels on the estimated value of $\varsigma_{k},$
null estimation is required by the conditionality principle \citep{conditionalityPrinciple1958},
in agreement with the analogy with conditioning on observed row or
column totals in contingency tables \citep{RefWorks:57}. See \citet{Shi2008458}
for further explanation of the relevance of the principle to estimation
of the null distribution. 

Accordingly, performance of each method of computing confidence levels,
whether under the assumed null distribution $\tilde{F}_{0}$ or estimated
null distribution $\hat{F}_{0}$, was evaluated in terms of the proximity
of $P_{k,i}^{x}\left(\left(\inf\Theta_{1},\theta^{\prime}\right);F_{0}\right),$
the confidence level of $\theta_{i}<\theta^{\prime}$ for trial $k$
and feature $i$ based on the null hypothesis of distribution $F_{0}\in\left\{ \hat{F}_{0},\tilde{F}_{0}\right\} ,$
to $P_{k,i}^{x}\left(\left(\inf\Theta_{1},\theta^{\prime}\right)|\varsigma_{k}=\sigma_{k}\right)$,
the corresponding true confidence level conditional on the realized
value $\sigma_{k}$ of $\varsigma_{k}$ used to generate the simulated
data of trial $k.$ For some $\alpha\in\left[0,1\right],$ the \emph{conservative
error} of relying on $F_{0}$ as the distribution under the null hypothesis
for the $k$th trial is the average difference in the number of confidence
levels incorrectly included in $\mathcal{B}=\left[\alpha,1-\alpha\right]$
and the number incorrectly included in $\bar{\mathcal{B}}=\left[0,1\right]\backslash\mathcal{B}:$
\begin{equation}
\sum_{i\in\mathcal{I}}\frac{1_{\mathcal{B}}\left(P_{k,i}^{x}\left(\Theta_{1}^{\prime};F_{0}\right)\right)1_{\bar{\mathcal{B}}}\left(P_{k,i}^{x}\left(\Theta_{1}^{\prime}|\sigma_{k}\right)\right)-1_{\mathcal{B}}\left(P_{k,i}^{x}\left(\Theta_{1}^{\prime}|\sigma_{k}\right)\right)1_{\bar{\mathcal{B}}}\left(P_{k,i}^{x}\left(\Theta_{1}^{\prime};F_{0}\right)\right)}{\left|\mathcal{I}\right|},\label{eq:conservatism}\end{equation}
where $\Theta_{1}^{\prime}=\left(\inf\Theta_{1},\theta^{\prime}\right)$
and where $\mathcal{I}=\left\{ 1,\ldots,9500\right\} $ for the unaffected
features or $\mathcal{I}=\left\{ 9501,\ldots,10^{4}\right\} $ for
the affected features. Here, $\alpha=1\%$ to quantify performance
near confidence values relevant to the inference problem of interpreting
the value of $P_{k,i}^{x}\left(\left(\inf\Theta_{1},\theta^{\prime}\right);F_{0}\right)$
as a degree of evidential support for $\theta_{i}<\theta^{\prime}.$
Values of the conservatism (\ref{eq:conservatism}) for the simulation
study described above appear in Fig. \ref{fig:mixtureCond}.

\begin{figure}
\includegraphics[scale=0.65]{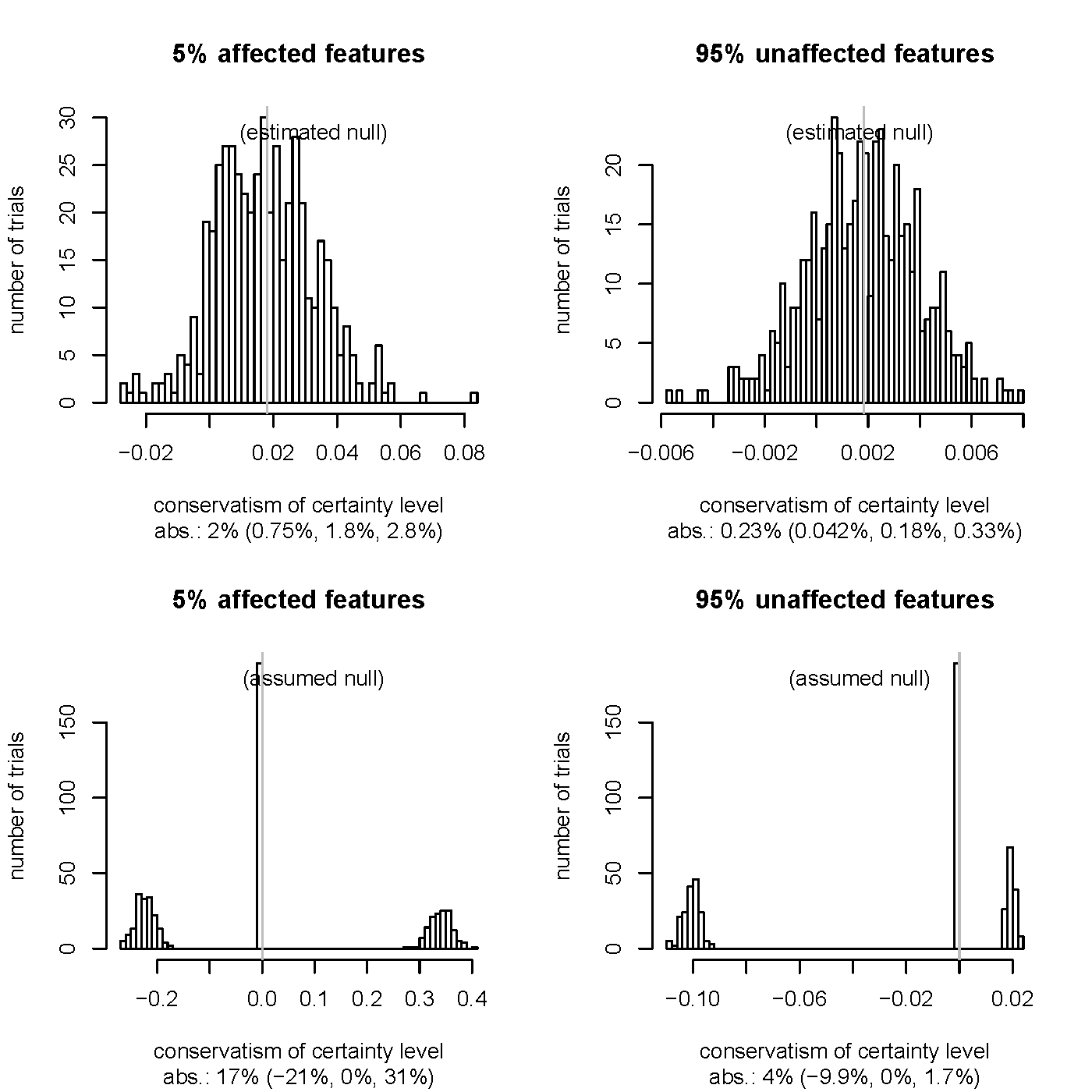}\caption{Conservative error (\ref{eq:conservatism}) when the assumed null
distribution is equal to the true null distribution conditional on
the most common value of the precision statistic $\left(\varsigma_{k}=1\right).$
The null distribution $F_{0}$ is the estimated distribution $\hat{F}_{0}$
in the top two plots and the assumed distribution $\tilde{F}_{0}$
in the bottom two plots. The two plots on the left and right give
the errors averaged over the 500 false and the 9500 true null hypotheses,
respectively.\label{fig:mixtureCond}}

\end{figure}

To determine the effect of analyzing confidence levels that are valid
marginal (unconditional) \emph{p}-values for the mixture distribution,
the confidence levels valid given $\varsigma_{k}=1$ were transformed
such that those corresponding to unaffected features are tail-area
probabilities under the marginal null distribution: \begin{eqnarray*}
P_{\theta^{\prime}}\left(Z_{k,i}<z_{k,i}\right) & = & \sum_{\sigma\in\left\{ 2/3,1,3/2\right\} }P\left(\varsigma_{k}=\sigma\right)P_{\theta^{\prime}}\left(Z_{k,i}<z_{k,i}|\varsigma_{k}=\sigma\right),\end{eqnarray*}
where $\Phi\left(z_{k,i}\right)$ or $P_{\theta^{\prime}}\left(Z_{k,i}<z_{k,i}\right)$
is the observed confidence level of $\theta_{k,i}<\theta^{\prime}$
before or after transformation, respectively. Fig. \ref{fig:mixtureMarg}
displays the results.

\begin{figure}
\includegraphics[scale=0.65]{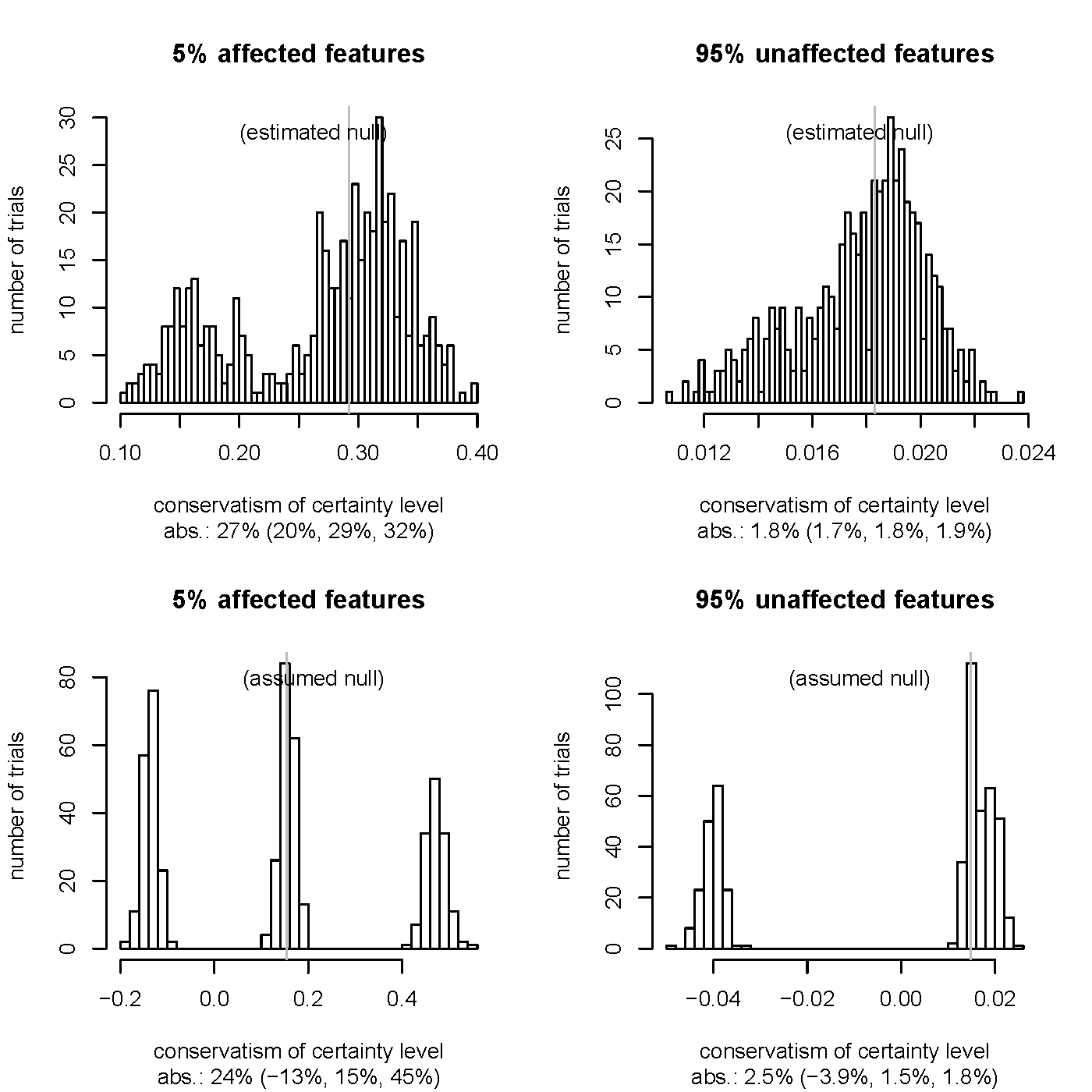}\caption{Conservative error (\ref{eq:conservatism}) when the assumed null
distribution is equal to the true null distribution marginal over
the distribution of precision statistic $\varsigma_{k}.$ The four
plots have the same arrangement as those of Fig. \ref{fig:mixtureCond}.\label{fig:mixtureMarg}}

\end{figure}

\subsection{\label{sub:Merit}Merit of estimating the null distribution}

While the degree of undesirable conservatism illustrates the potential
benefit of null estimation (§\ref{sub:Simulations}), it does not
provide case-specific guidance on whether to estimate the null distribution
for a given data set generated by an unknown distribution. Framing
the estimated null distribution as a conditioning statistic makes
such guidance available from an adaptation of a general measure \citep{Lloyd1992241}
that quantifies the benefit of conditioning inference on a given statistic.
Since an approximately ancillary statistic can be much more relevant
for inference than an exactly ancillary statistic, \citet{Lloyd1992241}
quantified the benefit of conditioning on a statistic by the sum of
its degree of ancillarity and its degree of relevance, each degree
defined in terms of observed Fisher information. To assess the benefit
of conditioning inference on the estimated null distribution, the
ancillarity and relevance are instead measured in terms of some nonnegative
divergence or \emph{relative information} $I\left(F||G\right)$ between
distributions $F$ and $G$ as follows. The ancillarity of the estimated
distribution $\hat{F}_{0}$ for $d_{1}$ affected features is the
extent to which the parameter of interest is independent of the estimate:
\begin{equation}
A\left(d_{1}\right)=-I\left(\hat{F}_{0}^{d_{1}}||\hat{F}_{0}\right).\label{eq:ancillarity}\end{equation}
Here, $\hat{F}_{0}^{d_{1}}$ represents the estimated null distribution
with its $d_{1}$ affected features replaced with unaffected features.
More precisely, $\hat{F}_{0}^{d_{1}}$ is the estimate of the null
distribution obtained by replacing each of the $d_{1}$ confidence
levels farthest from 0.5 with $\left(r-1/2\right)/d,$ the expected
order statistic under the assumed null distribution, where $r$ is
the rank of the distance of the replaced confidence level from 0.5.
Exact ancillarity, $A\left(d_{1}\right)=0,$ thus results only when
$\hat{F}_{0}^{d_{1}}=\hat{F}_{0},$ which holds approximately for
all $d_{1}$ if $\hat{F}_{0}$ is close to the assumed null distribution.
Conditioning on a null distribution estimate is effective to the extent
that its relevance, \begin{equation}
R=I\left(\hat{F}_{0}||\tilde{F}_{0}\right),\label{eq:relevance}\end{equation}
is higher than its \emph{nonancillarity}, $I\left(\hat{F}_{0}^{d_{1}}||\hat{F}_{0}\right)$. 

The importance of tail probabilities in statistical inference calls
for a measure of divergence $I\left(F||G\right)$ between distributions
$F$ and $G$ with more tail dependence than the Kullback-Leibler
divergence. The Rényi divergence $I_{q}\left(F||G\right)$ of order
$q\in\left(0,1\right)$ satisfies this requirement, and $I_{1/2}\left(F||G\right)$
has proved effective in signal processing as a compromise between
the divergence with the most extreme dependence on improbable events
$\left(\lim_{q\rightarrow0}I_{q}\left(F||G\right)\right)$ and the
Kullback-Leibler divergence $\left(\lim_{q\rightarrow1}I_{q}\left(F||G\right)\right).$
Another advantage of $q=1/2$ is that the commutivity property $I_{q}\left(F||G\right)=I_{q}\left(G||F\right)$
holds only for that order. The notation presents $I_{q}\left(F||G\right)$
as the order-$q$ \emph{information} gained by replacing $G$ with
$F$ \citep[§9.8]{RenyiBook1970}. Since the random variables of the
assumed and estimated null distributions are \emph{p}-values or confidence
levels transformed by $\Phi^{-1}$ (Fig. \ref{fig:CDFs}) and since
both distributions are normal, the relative information of order $1/2$
is simply\[
I_{1/2}\left(F||G\right)=-2\log_{2}\left(\frac{\left(\mu_{F}-\mu_{G}\right)^{2}}{4\left(\sigma_{F}^{2}+\sigma_{G}^{2}\right)}+\frac{1}{2}\ln\left(\frac{\sigma_{F}^{2}+\sigma_{G}^{2}}{2\sigma_{F}\sigma_{G}}\right)\right)\]
with $F=\N\left(\mu_{F},\sigma_{F}^{2}\right)$ and $G=\N\left(\mu_{G},\sigma_{G}^{2}\right).$

Assembling the above elements, the net \emph{inferential benefit}
of estimating the null distribution is\begin{equation}
B\left(d_{1}\right)=A\left(d_{1}\right)+R=I_{1/2}\left(\hat{F}_{0}||\tilde{F}_{0}\right)-I_{1/2}\left(\hat{F}_{0}^{d_{1}}||\hat{F}_{0}\right)\label{eq:benefit}\end{equation}
if there are $d_{1}$ affected features, where $\tilde{F}_{0}=\N\left(0,1\right)$
and where the ancillarity $A\left(d_{1}\right)$ and relevance $R$
are given by equations (\ref{eq:ancillarity}) and (\ref{eq:relevance})
with $I=I_{1/2}.$ Basing inference on the estimated null distribution
is effective to the extent that $B\left(d_{1}\right)>0.$ Fig. \ref{fig:benefit}
uses the gene expression data to illustrate the use of $B\left(d_{1}\right)$
to determine whether to rely on the estimated null distribution $\hat{F}_{0}$
or on the assumed null distribution $\tilde{F}_{0}$ for inference. 

\begin{figure}
\includegraphics[scale=0.7]{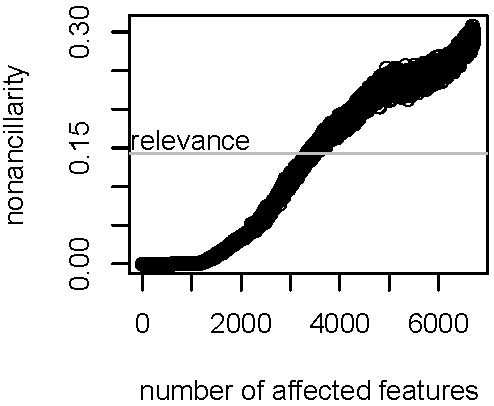}\caption{The nonancillarity $-A\left(d_{1}\right)$ versus the hypothetical
number $k_{1}$ of affected features. The gray horizontal line is
the relevance $R$ of null estimation and thus indicates the point
at which conditioning on the estimate goes from beneficial $\left(\left|A\left(d_{1}\right)\right|<R\right)$
to deleterious $\left(\left|A\left(d_{1}\right)\right|>R\right)$
according to equation (\ref{eq:benefit}). The data set, model, and
null distribution estimator are those of Example \ref{exa:tomato}
and Figs. \ref{fig:CDFs}, \ref{fig:microarray}, and \ref{fig:decisions}.\label{fig:benefit}}

\end{figure}

\section{\label{sec:Discussion}Discussion}

Whereas most adjustments for multiple comparisons are aimed at minimizing
net loss incurred over a series of decisions optimized over the sample
space rather than at weighing evidence in a particular data set for
a hypothesis, adjustments resulting from estimation of the distribution
of test statistics under the null hypothesis are appropriate for all
forms of frequentist hypothesis testing (§\ref{sub:Multiplicity}).
A form seldom considered in non-Bayesian contexts is that of making
coherent decisions by minimizing loss averaged over the parameter
space. Taking a step toward filling this gap, Section \ref{sub:Non-additive-loss}
provides a loss function suitable for genome-scale screening rather
than for confirmatory testing and illustrates its application to the
detecting evidence of gene upregulation or downregulation in microarray
data. 

Simulations measured the extent to which estimating the null distribution
improves conditional inference in an extreme multiple-comparisons
setting such as that of finding evidence for differential gene expression
in microarray measurements (§\ref{sub:Simulations}). While confidence
levels of evidence tended to err on the conservative side under both
the estimated and assumed null distributions, conservative error quantified
by numbers of confidence levels in $\left[1\%,99\%\right]$ compared
to the confidence levels conditional on the precision statistic $\varsigma_{k}$
was excessive under the assumed null but negligible under the estimated
null (Fig. \ref{fig:mixtureCond}). (Since the same pattern of relative
conditional performance was obtained by more realistically setting
$\log\varsigma_{k}$ equal to a variate that is independent and uniformly
distributed between $\log\left(1/2\right)$ and $\log\left(2\right),$
those results were not displayed.) Due to the heavy tails of the marginal
distribution of pre-transformed confidence levels under the null hypothesis,
transforming them to satisfy that distribution under the assumed null
increased their conditional conservatism, resulting in about the same
performance of estimated and assumed null distributions with respect
to the affected features. The case of the unaffected features is more
interesting: the assumed null distribution, which after the transformation
is marginally exact and hence valid for Neyman-Pearson hypothesis
testing, incurs 35\% more conservative error than the estimated null
distribution (Fig. \ref{fig:mixtureMarg}). Thus, the use of the marginal
null distribution in place of~ $\N\left(0,1\right),$ the distribution
conditional on the central component of the mixture, substantially
increases conservative error irrespective of whether the null is assumed
or estimated. These results suggest that confidence levels better
serve inductive inference when derived from a plausible conditional
null distribution than from the marginal distribution even though
the latter conforms to the Neyman-Pearson standard. This recommendation
reinforces the conditionality principle, which is appropriate for
the inferential goal of significance testing as opposed to the various
decision-theoretic motivations behind Neyman-Pearson testing (§\ref{sub:Multiplicity}).

Since the findings of the simulation study do not guarantee the effectiveness
of an estimated null distribution $\hat{F}_{0}$ over the assumed
null distribution $\tilde{F}_{0},$ Section \ref{sub:Merit} gave
an information-theoretic score for determining whether to depend on
$\hat{F}_{0}$ in place of$\tilde{F}_{0}$ for inference on the basis
of a particular data set. The score serves as a tool for discovering
whether the ancillarity and inferential relevance of $\hat{F}_{0}$
call for its use in inference and decision making.

\section{Acknowledgments}

This research was partially supported by the Faculty of Medicine of
the University of Ottawa and by Agriculture and Agri-Food Canada.
I thank Xuemei Tang for providing the fruit development microarray
data. The \emph{Biobase} \citep{RefWorks:161} and \emph{locfdr} \citep{RefWorks:208}
packages of R \citep{R2008} facilitated the computational work. 

\bibliographystyle{elsarticle-harv}
\bibliography{refman}

\end{document}